\documentclass[12pt,a4paper,titlepage,reqno]{amsart}
\usepackage{amsfonts}
\usepackage{latexsym}
\usepackage{graphicx}
\usepackage{amsmath}
\newtheorem{thm}{Theorem}
\newtheorem{cor}{Corollary}
\newtheorem{lem}{Lemma}
\newtheorem{prop}{Proposition}
\theoremstyle{definition}
\newtheorem{defn}{Definition}
\theoremstyle{remark}
\newtheorem{rem}{Remark}
\numberwithin{equation}{section}

\newcommand{\R}{\mathbb{R}}

\begin{document}

\thispagestyle {empty}

\title[]{Optimality of the free quantum evolution: the general case with nodes }
\author{\bf {Laura M. Morato}}

\dedicatory{ Retired professor,Universit\`a di Verona; E-mail:   morato.lauramaria@gmail.com}

\begin{abstract}
Extending and refining a preceding work,  the free evolution of a  quantum wave function, with fixed initial and final modulus, is considered in the general case where nodes are allowed. The method is based on considering a suitable simple functional related to the quadratic optimal transport cost in discrete time.  Under reasonable assumptions it is shown that, if nodes are present, the free evolution of the wave function with fixed initial and final modulus is a sort of local minimum.  If  no nodes are present the minimum is global, suggesting intrinsic instability of  the nodes in a free quantum evolution. 
\end{abstract}
\maketitle

%%% -------------------------------
\pagestyle{myheadings}
\thispagestyle{myheadings}
%\markboth{}
\markboth{L. M. Morato} {Free evolution}

%%% ------------------------------

\section{Introduction}

Consider  on the time interval $[0,1]$ the free Schroedinger equation for a finite dimensional system

\begin{equation}\label{Schroedinger}
i\partial_t \psi + \frac1 2 \nabla^2 \psi = 0
\end{equation}

\noindent and let $\psi $ be a solution in $L^2(\R^d, \mathbb C)$ such that $|\psi_o|^2 =\rho_o$ and $|\psi_1|^2=\rho_1$ for  some probability densities $\rho_o$ and $\rho_1$ on $\R^d$.  For the sake of simplicity we consider in this paper only the case where the solution is unique as an element of $L^2(\R^d, \mathbb C)$. 
We ask the question whether such an evolution is optimal, i.e. whether it corresponds to a minimum for a properly formulated optimisation problem. 

In a preceding work \cite {M2} it was shown that this is in fact the case when no nodes are present, and that the optimisation problem can be formulated in terms of  a generalisation of the quadratic cost functional of the Optimal Transport Theory, which can also be seen as an average of the classical action for a free particle (see \cite{V} section 8 and \cite{B.B.})

In \cite{M2} the result was made possible by exploiting the stochastic structure underlying the canonical formalism of Quantum Mechanics, for systems with a finite degrees of freedom, provided by Nelson's Stochastic Mechanics, introduced by Nelson in 1966 \cite {Nelson1}, where in particular  every solution of 1 is associated with a Markov diffusion process, with values in the configuration space, usually  called "Nelson's Diffusion".The Nelson diffusion associated with a wave function $\psi$ has a time dependent  probability density $\rho$  that is equal to $|\psi|^2$ and satisfies the fluid dynamics continuity equation with respect to a velocity field $v$ that is equal to the gradient of the phase.

     It is worth quoting that at the beginning of the history of Nelson's Stochastic Mechanics the occurrence of nodes in the wave function corresponded to a severe difficulty  from the mathematical  point of view:  in fact the presence of nodes gives rise to strong unboundedness in the drifts of Nelson diffusion, making difficult to prove the existence of the diffusion itself. 
      In the absence of nodes one can exploit the classical theorems on Stochastic Differential Equations with coefficients satisfying the Lipschitz and the sublinear growth  conditions, that ensure the existence of a unique strong solution. In  case  nodes are present the drifts are too singular. The problem was solved by Carlen in 1985 \cite{C1}. He introduced the concept of diffusions with "proper characteristics" and proved a weak existence result.      
     
The space of proper characteristics, denoted by $\mathcal P$,(see definition 1 in section 2) is given by a large set of fluid dynamics pairs $(\rho,v)$, where $\rho$ is the time dependent density of the flow and $v$ is the related velocity field. In this framework Carlen proved in great generality that all Nelson diffusions associated with solutions of a Schroedinger equation exist in a weak sense.
  
  The occurrence of nodes in the wave function has also generated a huge amount of literature concerning 
  the question whether Stochastic Mechanics, or other somehow similar theories, are correct or not from a physical point of view. The author wish to stress that the purpose of this work is only to establish a mathematical fact concerning the Schroedinger Equation for a free finite dimensional system. The reader can find a review on the above mentioned subject  for example in \cite{Reddiger}.
  \vskip5mm
  
 A difficult in extending  to the case with nodes the optimisation result given in \cite{M2}, by exploiting stochastic techniques, is that one needs to define all possible Nelson  diffusions  that are considered as defined on the same probability space, a property that is not guaranteed by the weak existence of the diffusions themselves.

  To overcome the problem we introduce a distinguished subset of the "proper  characteristics", denoted by $\Xi$, (see definition 2 in section 2), which in fact takes into account common examples that one can find in the Quantum Mechanics  text books, and we can prove that the diffusions with proper characteristics in $\Xi$ exist as strong solutions of a Stochastic Differential Equation.
    \vskip 10mm
   Introducing the non convex functional on $\mathcal P$, 
   
   $$
   A^Q(\rho,v):=\int _0^1 \int _{\R^d} (v^2- (\frac 1 2\nabla \ln \rho)^2)\rho dx dt, \quad (\rho,v)\in \mathcal P
   $$
   \vskip 10mm  
  
  \noindent and exploiting the good properties of Nelson's diffusions with characteristics in $\Xi$, (see proposition 1, proposition 2 and Lemma 2 in section 4),  we can prove that  a solution $  \psi$ of 1 with nodes is a sort of local minimum on the subset of pairs
   $(\rho,v)$ in $\Xi$ such that $\rho(0) = \rho_o$ and $\rho(1) = \rho_1$ while, in the assumption that $\psi$ is unique up to a constant phase, it is a global minimum if no nodes are present. 
   
   This fact suggests that the nodes  in a free quantum evolution are intrinsically unstable.

  The method relies on the construction of a convex functional which is a "representation" of $A^Q$ on a linear space of stochastic processes.
  \vskip5mm

 The paper is organized as follows:
 \vskip5mm 
 
  In section 2 we recall the definition of the space of proper (fluid dynamic) characteristics $\mathcal P$ as introduced by Carlen \cite{C1}. We then introduce the distinguished subset  $\Xi \in \mathcal P$   and  formulate the optimisation problem $P_1$.

 \vskip5mm 
 
 In Section 3 we state an extremality result  related to  $P_1$, which is a refinement  of a stationarity principle for Stochastic Mechanics that was given in the framework of the Stochastic Control Theory in \cite{GM} and \cite{Nelson2}. The refinement is done by exploiting the fluid dynamics formulation of it proposed by Loffredo \cite {L}.

\vskip 5mm
  
   In Section 4 we prove some properties of the diffusions with "proper characteristics" belonging to $\Xi$ and show in particular that they are strong solutions of a SDE so that all diffusions we are considering can be defined on the same probability space.

  \vskip 5mm
  In Section 5  we introduce a  suitable linear metric space of square integrable (non necessarily adapted and markovian) stochastic processes. A relevant subset of such a space is given by its markovian elements.
   We then consider  the elementary convex functional $F_n$ in discrete time (see (5.2)) and, exploiting  Nelson's renormalization formula (lemma 2), we are able to define an asymptotic convex functional
    $\hat F_\infty$  in continuous time. The restriction of this functional  on the set of markovian elements turns to be a "convex representation" of $A^Q$.
  \vskip10mm
  In Section  6 we prove the optimality result (Th 2) .

 \section{ The space of proper (fluid dynamics) characteristics and position of the optimisation problem }
 
 \vskip 5mm

 1)   {\bf The fluid associated with a smooth diffusion}
\vskip 5 mm

Let
\vskip 2mm 

- $b: R^d\times [0,1]\rightarrow \R^d$ be smooth and with sublinear growth. 
\vskip 2mm 
 
- $\rho_o$ be a smooth strictly positive probability density with finite variance on $\R^d$.
\vskip 2mm 
 
- $\xi_o$  be a $\R^d$-valued random variable that has probability density $\rho_o$.
\vskip 2mm

 Then, given an arbitrary probability space and a standard Wiener process on it, independent of $\xi_o$, there exists a unique solution to the SDE

$$
\xi_t = \xi_o + \int _0^t b(\xi_s,s)ds + W_t
$$

 Moreover $\xi_t$ has a probability density $\rho_t$ for all $t\in [0,1]$, which is a smooth strictly positive  solution of the Fokker-Planck equation

 $$
 \partial _t \rho= -\nabla\cdot(b\rho)+ \frac 1 2 \nabla^2 \rho
 $$
 \vskip 2mm
(Notice: one could also consider the case when the diffusion matrix is not equal to the identity).
\vskip 2mm

Defining

$$
v:= b- \frac {1}{ 2} \nabla \log \rho
$$

\noindent one can see that the pair $(\rho,v)$ satisfies the continuity equation

$$
\partial _t \rho - \nabla \cdot(\rho v)= 0
$$
\noindent that models the conservation of mass for a physical fluid with density $\rho$ and velocity field $v$ .

So, under  regularity assumptions on the coefficients and the assumption that the density is strictly positive, to any diffusion is naturally attached a flow in $\R^d$ with density $\rho$ and velocity field $v$.

 Carlen calls $(\rho,v)$ the "infinitesimal characteristic of the diffusion $\xi$". I will refer to it as the "fluid dynamics characteristic" or simply "characteristic" of $\xi$.

This observation (whose relevance was pointed out by Nelson in '66) was the starting point exploited by Carlen in order to prove the existence of a solution, in the Strook and Varadhan sense, for a wide class of SDE with unbounded drifts and constant diffusion matrix . Roughly, the idea is that a wide class of SDE with unbounded drifts have a  weak solution if they are constructed by means of a fluid dynamic characteristic with a "reasonable behaviour". Nelson's diffusions correspond, in great generality, to a subset of such a class.
\vskip 10mm

  We will use of the following notation:

\begin{equation}\label {notation}
b[\rho,v]:= v+ \frac {1}{ 2 }\nabla \ln \rho
\end{equation}

\vskip 10mm

2) { \bf The space of proper (fluid dynamic) characteristics and a distinguished subset of it}
\vskip10mm

Carlen defines the "set of proper characteristics " $\mathcal P$ as follows \cite{C3}:

\begin{defn}
 We say that $(\rho,v)$ belongs to $\mathcal P $ if $\rho$ is a continuously depending on time probability density on $\R^d$ and  $v$ is a measurable $d$-dimensional time dependent vector field on $\R^d$, such that the "finite action condition"  and the continuity equation, in the sense of distributions, hold, i.e.

\begin{equation}\label{finite energy}
\int _0^1 \int _{\R^d} (v^2+ (\frac 1 2\nabla \ln \rho)^2)\rho dx dt < \infty\\
\end{equation}

\noindent and

\begin{equation}
\partial _t\rho + \nabla\cdot(\rho v)=0 \quad (\text{distrib.})\\
\end{equation}
\end{defn}

\vskip 10mm

Motivated by typical examples of eigenfunctions with nodes that one can find in Quantum Mechanics  text books, we restrict our attention to the following subset $\Xi$
of the set of proper characteristics $\mathcal P$.

Putting  

 $$
 N_{\rho} := \{(x,t)\in \R^d\times [0,1] : \rho(x,t)=0     \}
 $$
 
 \noindent we give the following definition
\vspace{5mm}

\begin{defn}

$\Xi $:= $\{$     $(\rho,v)$ that belong to $\mathcal{P}$  and are such that
\vspace {2mm}

h1)   If it is not empty, the set $N_\rho$ is the finite union of isolated points or regular  curves or (iper)surfaces in $\R^d\times [0,1]$ such that $N^c_\rho$ is connected or the finite union of open connected sets.

\vspace{2 mm}
h2)  $\rho$ is of class $C^ {\infty} $ on $\R^d \times [0,1]$ and  the restriction on $N^c_\rho$ of $v$ is also of class $C^{\infty}$ .

\vspace{2 mm}
h3) The limit of $\rho(x,t)v(x,t)$ for $(x,t)$ going to the border of $N_\rho$, with $(x,t)\in N_{\rho}^c$, is equal to zero.$\}$
\end{defn}

\vspace{2mm}
As far as the free evolution is concerned, the simplest way of constructing  examples with nodes is that of considering as initial wave function one of the above mentioned examples.
 \vskip 2mm

{\bf Example}   We consider the solution of the one dimensional free Schroedinger equation with the first excited state of the one dimensional harmonic oscillator as initial state, i.e.

$$i \partial_t \psi (x,t)+\frac 1 2  \partial _{xx}\psi(x,t) = 0\quad,\quad \psi(x,0) = \psi_0(x)$$

$$\psi_0 (x) = C(0) x  e^{-x^2/4}$$

After some manipulations one finds

$$\psi (x,t) = C(t) x e^{ -x^2/{4 r(t)}} e^ {i t x^2/{8 r(t)}}$$

\noindent where $C(t)$  and $r(t)$ are time dependent positive constants.

Then we have

\vskip 2mm

--$\rho(x,t) = |C(t)|^2 x^2 e^{- x^2/ 2 r(t)}$

\vskip 2mm
-- $N_\rho = \{(x,t) :\quad x = 0, \quad t \in [0,1]\}$.
\vskip 5mm
For all $(x,t)\in N_\rho^c $ the drift is  

$$
b[\rho,v](x,t):= v(x,t)+ \frac 1 2 \nabla \ln \rho(x,t)
$$

\noindent where

$$ \frac 1 2 \nabla ln \rho (x, t) = \frac 1 x - \frac x {2 r(t)} $$

$$v(x, t) = \frac {t x}{4 r (t)} $$

\vskip 10mm

3) {\bf Position of the problem}
\vskip 10mm
We consider on the space of proper characteristics $\mathcal P$ the functional

\begin{equation}
A^Q(\rho,v):=
\int _{\R^d}\int_0^1 (v^2-(\frac 1 2 \nabla ln \rho)^2)\rho dt dx
\end{equation}

This functional can be seen as a generalisation of that considered in the framework of the Differential Point of Vew in Optimal Transportation Theory  ( see \cite{B.B.} and \cite{V} Section 8), or the regularized limit of a mean discretized classical action for a diffusive motion (see Section 5).
\vskip 5mm

We denote by  $\Xi(\rho_o,\rho_1)$ the set of pair $(\rho,v)$ which belong to $\Xi$ and such that
 $\rho(0)=\rho_o$ and  $\rho(1)=\rho_1$.

The optimisation problem that we are considering is
\vskip 10mm

$P_1$

\begin{equation}
  \underset {(\rho,v)\in \Xi(\rho_o,\rho_1)} {min}A^Q(\rho,v)
\end{equation}
\vskip 10mm

\section{A sufficient and necessary extremality condition}

Here we give a refinement of the fluid dynamics version, proposed in \cite{L}, of the stationarity principle for Stochastic Mechanics given, in the framework of the Stochastic Control Theory, 
 in \cite {GM} and \cite{Nelson2}.
  
  Handling the nodes of the wave function is not easy in the framework of the stochastic control  theory. On the contrary it turns to be  natural in the fluid dynamics approach.
  
 (The following lemma  can be easily extended to the case of a system subject to a scalar or electromagnetic potential).

\begin{lem}

Let  $(\delta \rho,\delta v)$ satisfy the following conditions
\vskip 2mm
 a) $\delta \rho $ belongs to $\in C_K^\infty (N_\rho^c , \mathbb{R})$ and its support is in $N_\rho^c$. 
 
 b) $\int_{\R^d} \delta \rho dx $ is equal to zero.
 
c)  $\delta v$ belongs to $  C_K ^\infty(N_\rho^c\ , \R ^d)$ and its support is in $N_\rho^c$. Moreover the pair $(\rho+y\delta \rho, v+y\delta v)$ satisfies the continuity equation in a neighborhood of $y=0$, i.e., for all $(x,t)\in N_\rho^c$,  

 $$
 \frac {d}{dy}\{\partial _t (\rho +y\delta \rho)+\nabla \cdot [(\rho+y\delta\rho)(v+y\delta v)]\}(x,t) |_{y=0} =0     
$$

\vskip 2mm
Putting

 \begin{equation}\label{D}
\ D A^Q(\rho,v) (\delta\rho,\delta v) := \underset{\epsilon \to 0^+}{lim} \frac 1 \epsilon [A^Q (\rho + \epsilon \delta \rho, v + \epsilon \delta v)-A^Q(\rho , v ) ]
\end{equation}

\noindent a sufficient condition in order that an element $(\rho,v)$  of 
    $\Xi (\rho_0,\rho_1)$ satisfies the equality
\begin{equation}\label{critical}
 D A^Q (\rho, v) (\delta\rho, \delta v) = 0 , 
\end{equation}
\noindent $ \forall (\delta \rho, \delta v)$ such that a), b) and c) are satisfied, is that the following holds
\vspace {5mm}

i) (Local gradient condition) There exists   $S \colon N_\rho^c \to  \R$  such that, for all
 $(x,t)\in N_{\rho}^c$,
 
 \begin{equation} \label{Gradient}
v(x,t) = \nabla S(x,t)
\end{equation}

ii) The pair $(\rho, S) $ satisfies the Hamilton Jacobi Madelung's equation on $N_\rho^c$ i.e., for all $(x,t)\in N_\rho^c$,

\begin{equation}\label{HJM}
\partial_t S(x,t) + \frac 1 2 (\nabla S(x, t))^2 -
\frac 1 2 \frac {\nabla^2 \sqrt{\rho(x,t)}}{\sqrt{\rho(x,t)}} = 0
\end{equation}
\vskip 5mm

If $v(x,t)$ is different from zero on $N_\rho^c$  then the condition is also necessary.

\end{lem}

\begin{proof}

1) Sufficient condition
\vskip 5mm

\vspace{2mm}

Let $\lambda : N_\rho^c  \to  \R$ differentiable. Define, for every $(\rho,v)\in \Xi$, the extended action

$$
\bar A (\rho,v,\lambda) := A^Q (\rho, v) + \int \int_{\R^d\times [0,1]} \lambda (\partial_t \rho + \nabla\cdot (\rho v)) dx dt
$$

 For every $(\rho,v)$ in $\Xi $, one has

$$
\bar A(\rho,v,\lambda) := A^Q (\rho, v)
$$
\noindent and, for all $(\delta \rho,\delta v)$ that satisfy a), b) and c)
$$
 D \bar A(\rho, v, \lambda) (\delta \rho, \delta v, 0) =  D A^Q (\rho, v)(\delta \rho, \delta v)
$$
\vskip 2mm

We observe that we can write, for some integer $m$, 

$$
N_{\rho}^c := \bigcup _{i=1}^{m} N_i^c
$$

\noindent where, for all  $i= 1,\dots,m $, $N_i^c$, is a  connected region in $\R^d\times [0,1]$.  The extended action reads, being $N_\rho$ of zero Lebesgue measure on $\R^d\times [0,1]$,

$$ 
\bar A (\rho,v,\lambda)=\sum _{i=1}^{m} \int_{N_i^c} (v^2-(\frac 1 2 \nabla ln \rho)^2)\rho + \lambda (\partial_t \rho + div (\rho v))] dx dt
$$

 Calculating the derivative  of $\bar A$, as defined by \eqref{D}, and integrating by parts each term,  we get

$$
 D \bar A (\rho, v, \lambda) (\delta\rho, 0, 0) =
 $$
 
\begin{equation}\label{E1}
= \sum _{i=1}^{m} \int_{N_i^c} [-\partial_t \lambda -( \nabla\lambda) v + v^2 + \frac {\nabla^2 \sqrt\rho}{\sqrt \rho}]\delta \rho dx dt 
\end{equation}

\noindent and
$$
 D \bar A(\rho, v, \lambda) (0, \delta v, 0) =
$$
\begin{equation}\label{E2}
= \sum _{i=1}^{m} \int_{N_i^c}  (2 v - \nabla \lambda) \rho \delta v dx dt 
\end{equation}
(The calculation is very easy exploiting the change of variables $\rho:= \exp 2R$).

Then if \eqref{Gradient} and \eqref{HJM} are satisfied for every $(x,t)\in N_{\rho}^c$ and choosing $\lambda : = 2 S $ we have 
$$
D A^Q (\rho,\nabla S)(\delta \rho,\delta v)= D \bar A(\rho, \nabla S, 2 S) (\delta\rho, \delta v, 0) = 0
$$

\noindent for all $(\delta \rho,\delta v)$ which satisfy a), b) and c).

\vskip 10mm

2) Necessary condition
\vskip5mm 

We assume that for some $(\rho, v)\in \Xi$, with $v(x,t)$ different from zero for every $(x,t)$ belonging to $N_\rho^c$, we have

 \begin{equation}\label {C_o}
  D A^Q (\rho, v) (\delta\rho, \delta v) = 0 
 \end{equation}
 
 \noindent for all $(\delta \rho,\delta v)$ satisfying a), b) and c).

 Let $(x^*,t^*)$ belong to  $N_\rho^c$. We choose a $\delta\rho^* $ , satisfying a) and b) such that  its support is a neighborhood $I_{(x^*,t^*)}$  of $(x^*,t^*)$ . We assume that $I_{(x^*,t^*)}$ is a subset of $N_\rho^c$ and that $\delta\rho^* $ is positive in the point $(x^*,t^*)$ and has there its maximum. 
 
 To find a corresponding $\delta v^*$ we observe that condition c) gives the linear equation in the unknown $\delta v^*\rho$. 

\begin{equation}\label {cont inf}
\nabla\cdot (\delta v^*\rho)= -[\partial_t \delta \rho^* +\nabla\cdot (\delta \rho^* v)]
\end{equation}

Then, after a simple manipulation and recalling  a) and b), one can see that there exists $\delta v^*$ such that $(\delta \rho^*,\delta v^*)$ satisfies all the assumptions of the lemma.
\vskip 5mm

We observe that, by \eqref{cont inf},  if $\delta\rho^* $ reaches its maximum at $(x^*,t^*)$, then $\partial_t \delta \rho^*(t,x)$ is equal to zero in $t=t^*$ and, since  $v(x^*,t^*)$ is different from zero by assumption, one can choose $\delta v^ * (x^*,t^*)$ also different from zero so that $\delta v^*$ is different from zero in a neighborhood of $(x^*,t^*)$. 
  
   By the assumption we have
 \begin{equation}\label {C}
  D A^Q (\rho, v) (\delta\rho^*, \delta v^*) = 0 
 \end{equation}
 
\noindent so that, for any differentiable function  $\lambda :N_\rho^c\to \R$,
 
$$
  D\bar A^Q (\rho, v,\lambda) (\delta\rho^*, \delta v^*,0) = 0 
 $$
 
 \noindent and, in particular
 
 $$
  D\bar A^Q (\rho, v,\lambda) (\delta\rho^*,0,0) = 0 
 $$

\noindent and

  $$
  D\bar A^Q (\rho, v,\lambda) (0, \delta v^*,0) = 0 
 $$

 Then, recalling \eqref{E1} and \eqref{E2}, the fundamental argument in calculus of variations gives

\begin{equation}\label{N1}
 -\partial_t \lambda (x^*,t^*)- \nabla \lambda(x^*,t^*) v(x^*,t^*) + v(x^*,t^*)^2 + \frac {\nabla^2 \sqrt\rho(x^*,t^*)}{\sqrt \rho(x^*,t^*)}=0
 \end{equation}
 \vskip 10 mm
 
  \noindent and
  \begin{equation}\label {N2}
  (2\nabla\lambda(x^*,t^*) -v(x^*,t^*))\rho(x^*,t^*) = 0
 \end{equation}
 
 By the arbitrariness of $(x^*,t^*)\in N_\rho^c$  equalities  \eqref{N1} and \eqref{N2} hold on every point of $N_\rho ^c$.
 
 Then, denoting $2 \lambda$ by $S$, necessarily we have both  the local gradient condition
 
 $$
 v(x,t)= \nabla S(x,t) ,\quad   \forall (x,t )\in N_\rho^c
$$

 \noindent and the HJM equation on $N_\rho^c$
 
$$
\partial_t S(x,t) + \frac 1 2 (\nabla S(x, t))^2 -
\frac 1 2 \frac {\nabla^2 \sqrt{\rho(x,t)}}{\sqrt{\rho(x,t)}} = 0,\quad   \forall (x,t )\in N_\rho^c
$$

  \end {proof}

\vskip 5mm

\begin{rem}

Let $\tilde\delta\rho$ have  a support that is not a subset of $N_\rho^c$. We have, expliciting the boundary  term,

$$
 D \bar A (\rho, v, \lambda) (\tilde \delta\rho, 0, 0) =
 $$
 
$$
= \sum _{i=1}^{m} \int_{N_i^c} [-\partial_t \lambda - \nabla (\lambda) v + v^2 + \frac {\nabla^2 \sqrt\rho}{\sqrt \rho}]\tilde \delta \rho\} dx dt +$$

$$
 + \sum _{i=1}^{m} \int_{N_i^c} \{\partial_t(  \lambda \tilde \delta \rho) + \nabla\cdot( \lambda v\tilde \delta\rho) \}dx dt 
$$

Then,  by the divergence theorem, the boundary term goes not to zero because $ \tilde \delta \rho$ does not vanish on $N_\rho$ .Thus
  $D \bar A (\rho, v, \lambda) (\tilde \delta\rho, 0, 0)$  cannot be equal to zero.

\end{rem}

\vskip 20mm

\begin{cor}

If $\psi$ is a solution in $L^2(\R^d , \mathbb C)$ of the Schroedinger equation

\begin{equation}\label{Schroedinger1}
i\partial_t \psi + \frac1 2 \nabla^2 \psi = 0, \quad |\psi_o|^2 =\rho_o, \quad |\psi_1|^2 =\rho_1
\end{equation}

\noindent and $(\rho,\nabla S)\in \Xi(\rho_o,\rho_1)$, then 
\begin{equation} 
D A^Q (\rho, \nabla S) (\delta\rho, \delta v) = 0
\end{equation}

\noindent for all $(\delta\rho, \delta v)$ satisfying a), b and c).

Viceversa, if

\begin{equation}\label{A2}
D A^Q (\rho, v) (\delta\rho, \delta v) = 0
\end{equation}

 \noindent for some $(\rho,v)\in \Xi _{\rho_o,\rho_1}$  such that $v$ is different from zero on $N_\rho^c$ and for all $(\delta\rho, \delta v)$ satisfying a), b) and c), then there exists a function $S:N_\rho^c\times [0,1]\to \R$ such that  $v=\nabla S$ on $N_{\rho}^c$, and
 $\psi:= \rho^{\frac 1 2}exp^{iS}$ is a  solution of the Schroedinger equation \eqref{Schroedinger2}
 in $L^2 (R^d\to \mathbb C)$.

\end{cor}

\begin{proof}

Putting 
 
$$
\psi := \rho^{\frac 1 2}\exp ^{iS}
$$

\noindent we assume that equation \eqref {Schroedinger1}  is satisfied and that  $(\rho,\nabla S)\in \Xi _{\rho_o,\rho_1}$

Then, by Lemma 1,  for all $(x,t)\in N_\rho^c$, \eqref{Schroedinger1} $\rho$ and $S$ satisfy the system of equations

\begin{equation}\label{cont2}
\partial_t \rho(x,t) + \nabla\cdot (\rho(x,t) \nabla S(x,t) )= 0
\end{equation}

\begin{equation}\label{HJM2}
\partial_t S(x,t) + \frac 1 2 (\nabla S(x, t))^2 -
\frac 1 2 \frac {\nabla^2 \sqrt{\rho(x,t)}}{\sqrt{\rho(x,t)}} = 0
\end{equation}

\noindent with the condition $|\psi_o|^2 =\rho_o, \quad |\psi_1|^2 =\rho_1$. 

  The first  is the continuity equation and the second is the Hamilton Jacobi Madelung equation.

Putting $v(x,t):=\nabla S(x,t)$ for all $(x,t)\in N_\rho^c$, both the local gradient condition \eqref{Gradient} and the HJM equation \eqref{HJM2}  are satisfied. Then the sufficient condition in Lemma 1 holds, so that
 
 $$
D A^Q (\rho, \nabla S) (\delta\rho, \delta v) = 0
$$

\noindent for all $(\delta\rho, \delta v)$ satisfying a) and b).
\vskip 10mm

Conversely if \eqref{A2} is satisfied by some $(\rho,v)\in \Xi_{\rho_0,\rho_1}$ with $v$ different from zero, thus the necessity condition in Lemma 1 says that, for some  $S: N_\rho ^c\to \R$ , the current velocity $v$ must satisfy the equality

$$
v(x,t):= \nabla S(x,t) 
$$

\noindent for all $(x,t)\in  N_\rho^c$, and

$$
\partial_t S(x,t) + \frac 1 2 (\nabla S(x, t))^2 -
\frac 1 2 \frac {\nabla^2 \sqrt{\rho(x,t)}}{\sqrt{\rho(x,t)}} = 0  .
$$

\vskip 20mm

Introducing $\psi := \rho ^{\frac 1 2}\exp i S $, since $N_\rho$ has Lebesgue measure equal to zero and $\int _{R^d\times [0,1]} \psi \psi^* dx dt =1$, 
$\psi := \rho ^{\frac 1 2}\exp i S $  defines a solution in $L^2(\R^d\times [0,1], \mathbb C)$ of the Schroedinger equation \eqref{Schroedinger1}.

  \end{proof}

\vskip 20mm
\section { Digression: some  properties of the diffusions with characteristics belonging to $\Xi$}

The following theorem holds for any pair in $\mathcal P$ \cite{C1}.

\begin{thm}(Carlen)
 Let $\Omega$ be the set of continuous functions from $\R$
to $\mathbb{R}^d$,  $\mathcal F$ denote the associated Borel $\sigma$-algebra
and the filtration $(\mathcal F_t)_t$ be defined in the natural way. Let also $X$ denote the configuration process.

Let the pair $[\rho,v]$ belong to $\mathcal P$ and $b[\rho,v]$ be defined as 
$$
b[\rho,v]:= v+ \frac 1 2 \nabla \ln \rho
$$
\vskip 5mm

Then there exists a (unique)
probability measure $\mathbb P$ on
$(\Omega,\mathcal F,(\mathcal F_t)_t)$,
such that $X$ satisfies $\mathbb P$ a.s. the equality

\begin{equation}
X_t - X_0 - \int_0^t b[\rho,v] (X_s, s) ds =  W_t
\end{equation}

\noindent where $W$ is a standard Brownian Motion on $(\Omega,\mathcal F,(\mathcal F_t)_t),
\mathbb P)$.
\vskip 2mm

Moreover $X_t$ has a probability density equal to $\rho_t$. 

\end{thm}

\vskip 10mm

 In order to study the subset of Carlen's processes that we are considering, it is important  to recall the property of non attainability of nodes for Nelson's diffusions.
This problem was firstly studied  by Nelson \cite{Nelson2} and later revisited by other authors.We refer the reader to the general result proved  by Zeng \cite{Zeng}.

 \begin{prop}(Non-attainability of the zeroes of the density)
 
  Let $(\Omega,\mathcal F,(\mathcal F_t)_t)$  be defined as in Theorem 1 and let $X$ be the configuration process.
 Let  also $\mathbb P$ be a probability measure on $(\Omega,\mathcal F,(\mathcal F_t)_t)$   and $W$ a standard Brownian Motion on $(\Omega,\mathcal F,(\mathcal F_t)_t,\mathbb P)$.

Let $b:=b[\rho,v]$ be such that 
$(\rho,v)$ belongs to $\Xi$ and let $\xi$  satisfy the stochastic differential equation

$$
\xi_t = X_o +\int _o^t b(\xi_t,t)dt + W(t)
$$

Then
 
$$\mathbb P \{\omega \in \Omega : (\omega(t),t)\in N_{\rho}^c\quad  \forall t\in [0,1]   \}=1
 $$

\end{prop}

\vskip 5mm

In the following we will denote by SDE $(X_o,b,\mathbb I )$  the stochastic differential equation with initial condition $X_o$, drift field $b$ and diffusion coefficient equal to the identity matrix in $\R^d$.  We prove the  existence of a  strong solution in the case where $b$ is equal to $b[\rho,v]$ whith $(\rho,v)$ belonging to
 $\Xi$.

\begin{prop}
Let $b:= b[\rho,v]$  with $(\rho,v)\in \Xi$.
 Then the S.D.E. $(X_o,b,\mathbb I)$ has a strong solution.
 \end{prop}

 \begin{proof}
 
 Let $(X,\mathbb P,W)$ be Carlen's weak solution of the S.D.E. $(X_o,b,\mathbb I)$, so that
 
 \begin{equation}\label{uno}
 X(\omega,t)=X_o(\omega) + \int _0^1 b(X(\omega,s),s) ds +W(\omega,t), \mathbb P a.s.
 \end{equation}

Let  $(D_j)_j$ be an increasing sequence of compact subset of $\R^d \times [0,1]$ converging to $\R^d \times [0,1]$ and $ ( \Phi_j )_j $ be defined by

\begin{equation}\label{Phi}
\Phi_j (x,t) = \begin{cases}
                  1     \qquad     (x,t) \in D_j \\
                  0     \qquad     (x,t) \in \underset{m=j+2}{U^\infty} (D_m \setminus D_{m-1})
               \end{cases}
\end{equation}

 Then $b\Phi_j$ is of class $C_b^\infty$ and

$$
(b\Phi_j)(x,t) = b(x,t)\quad \forall (x,t)\in D_j
$$

  The S.D.E. $(X_o,b\Phi_j,\mathbb I)$ has a strong solution and we can define the "regularized process" $X _j$ by

\begin{equation}\label{due}
X_j(\omega,t)=X_o(\omega) + \int _0^t (b\Phi_j)(X_j(\omega,s),s) ds +W(\omega,t), \mathbb P a.s.
\end{equation}

Defining

$$
\Omega_j :=\{ \omega :(X(\omega,t),t)\in D_j\quad \forall t\in [0,1]\}
$$

\noindent we have from \eqref{uno} and the definition of $\Phi_j$, for all $\bar\omega$ in $\Omega_j$

\begin{equation}\label{due'}
X(\bar \omega,t)=X_o(\bar \omega) + \int _0^t (b\Phi_j)(X(\bar\omega,s),s)ds +W(\bar\omega,t)
\end{equation}

Then for the uniqueness of the path-wise solution of equation \eqref{due} we have 

$$
X(\bar\omega,t)=X_j(\bar \omega,t), \forall t\in[0,1], \forall \bar\omega\in \Omega _j
$$

\noindent so that

\begin{equation}\label{contiene}
\{ \omega :(X_j(\omega,t),t)\in D_j \quad  \forall t\in [0,1]\}\supseteq \Omega_j
\end{equation}
\vskip 10mm

Let now $(\mathbb P',W')$ be arbitrarily chosen on $(\Omega,\mathcal F,(\mathcal F_t)_t)$ .
Since the S.D.E. $(X_o,b\Phi_j,\mathbb I)$ has strong solution, there exists $\xi_j$ satisfying the equality

\begin{equation}\label{tre}
\xi_j(\omega,t)=X_o(\omega) + \int _0^1 (b\Phi_j)(\xi_j(\omega,s),s) ds +W'(\omega,t), \mathbb P' a.s.
\end{equation}

\noindent and, by unicity in law of the solutions of the SDE $(X_o,b\phi_j, \mathbb I)$, $\xi_j$ has the same law as  $ X_j$.

\vskip 5mm

Introducing

$$
\Omega '_{jj} := \{\omega : (\xi_j(\omega,t),t)\in D_j, \quad \forall t \in [0,1] \}
$$

\noindent we define 

$$
\xi(\omega,t):=\xi_j(\omega,t), \quad \forall \omega \in \Omega'_{jj}
$$

By the localization theorem applied to the process $(\xi_j(\omega,t),t)$,we have

$$
\xi_j(\omega,t)=\xi_{j'}(\omega,t), \quad \forall \omega \in \Omega'_{jj},\quad j'>j,
$$

Then, $\mathbb P'a.s.$ for $\omega'\in \Omega'_{jj}$ and for all $j\in \mathbb N $ we define

$$
\xi(\omega',t)= \xi_j(\omega',t)=X_o(\omega') + \int _0^1 (b\Phi_j)(\xi_j(\omega',s),s) ds +W'(\omega',t)=
$$

$$
=X_o(\omega') + \int _0^1 b(\xi_j(\omega',s),s) ds +W'(\omega',t)=
$$

Since $\xi_j$ and $X_j$ are equal in law and recalling \eqref{contiene}

$$
\mathbb P'(\Omega'_{jj}) \equiv \{\omega : (\xi_j(\omega,t),t)\in D_j, \quad \forall t \in [0,1] \}=
$$

$$
=\mathbb P \{\omega : (X_j(\omega,t),t)\in D_j, \quad \forall t \in [0,1] \}\geq \mathbb P (\Omega_j)
$$

Finally, recalling the property of the non attainability of nodes (Prop 1),

$$
\lim _{j\to \infty} \mathbb P ({\Omega_j})=\mathbb P\{ \omega : (X(\omega,t),t) \in N_{\rho}^c \forall t\in [0,1]\}=1
$$

\noindent one also has

$$
\lim _{j\to \infty} \mathbb P'( {\Omega'_{jj}}) =1
$$

Concluding we can write
$$
\xi(\omega,t)=X_o(\omega) + \int _0^1 b(\xi(\omega,s),s) ds +W'(\omega,t), \quad \mathbb P' a.s.
$$

\end{proof}

\vskip 10mm

 \begin{prop} 
Let $(\rho_y,v_y)$ belong to $\Xi$ for all $y \in  [0,1]$ and  have uniformly bounded derivative with respect to $y$.

 Then the process $X(y, t) := q_t^{b[\rho_y,v_y]}$ depends continuously on $y$, can be differentiated with respect to $y$ at all orders and all derivatives have finite momenta.
   
\end{prop}
\begin{proof}
We consider firstly the case where no nodes are present. Then, if $(\rho_y,v_y)$ have uniformly bounded derivative with respect to $y$, then one can see that $b[\rho_y,v_y]$ is by construction of class $C_b^\infty$ and the derivatives with respect to $y$ exist and are uniformly bounded.
Then one can apply the theorem of smooth dependence on parameters for solutions of a S.D.E. so that the process
  $X(y, t) := q_t^{b[\rho_y,v_y]}$ depends continuously on $y$, can be differentiated with respect to $y$ at all orders and all derivatives have finite momenta.
 
 In the case where  $(\rho_y,v_y)$ is a generic element of $\Xi$ for all $y$, we consider the sequence
 of functions $\Phi_j$ as defined in the proof of proposition 2 and introduce the process
 
 $$
 X^j(y,t):= q_t^{\Phi_j b[\rho_j,v_j]}
 $$
 
 Then we can exploit again the theorem on smooth dependence on parameters for all $j\in \mathbb N$.  Following the same argument as in the end of the proof of proposition 2, the sequence  $\{X^j(y,t)\}_j$ defines asymptotically a process ${X(y,t)}$ with the desired properties. 
 \end{proof}

\vskip10mm
Finally we check that Nelson's celebrated renormalisation formula \cite{Nelson*} holds for diffusions with characteristics in $\Xi$.
\vskip 10mm
\begin{lem}

Let $b\equiv b[\rho, v], (\rho,v)\in \Xi$.

Then

\begin{equation}
\lim_{n\rightarrow\infty} n \mathbb{E} \sum_{i=0}^{n-1} (\triangle q_i^b)^2 - nd =
\int_0^1 \int_{\mathbb{R}^d}
 (v^2 - (\frac 1 2 \nabla ln \rho)^2)\rho dx dt <\infty  
\end{equation}

\end{lem}

\begin{proof} The proof is essentially the same as in reference \cite{Nelson*}. We report it in detail for reader's convenience.

Define
$
\Omega_o := \{\omega : q_t^b (\omega) \in N_\rho^c \; \forall t \in [0,1]\}
$

We can easily calculate the path-wise expansion of $b(q_s^b,s)$ in Taylor series and insert it in the expression
$$
\mathcal X_{\Omega_o} \Delta q_t^b\equiv \mathcal X_{\Omega_o}( \int_{t_i}^{t_{i+1}} b(q_s^b, s) ds +
\triangle W_{t_i}).
$$
One can see that
$$
b(q_s^b, s) = b(q^b_{t_i},t_i)+ \nabla b(q^b_{t_i},t_i) (q_s-q_{t_i}) +O(q_s-q_{t_i} )^2=
$$
$$
=b(q^b_{t_i},t_i)+ \nabla b(q^b_{t_i},t_i) (W_s-W_{t_i}) +\mathcal A _i(s,t_i)
$$

The order of $\mathcal A _i(s,t_i)$ with respect to $(s-t_i)$ can be estimated by  espliciting the lagrangian rest and by recalling that the paths of a  the Brownian Motion $W$ satisfy the Holder condition

$$
|W_s-W_t|<  |s-t|^\gamma , \quad \forall \gamma \in [0,\frac 1 2 )
$$

Putting for example $\gamma :=\frac 1 4$ one finds that $\mathcal A_i(s-t)$ is $o(s-t)^{\frac 3 2}$.

\noindent We find

\begin{align*}
\mathbb{E}\mathcal X_{\Omega o} (\triangle q_{t_i}^b)^2 =&
\mathbb{E}  \mathcal X_{\Omega _o} \{\Delta W_{t_i} + b(q_{t_i, t_i}^b)\frac 1 n +\\
+& \nabla b (q_{t_i}^b, t_i) \int_{t_i}^{t_{i+1}}(W_s - W_{t_i}) ds +
o(\frac 1 n )^{\frac 3 2}\}^2
\end{align*}

Since by Proposition 1 one has that  $\mathbb{P}(\Omega_o)$ is equal to $ 1$, we can drop $\mathcal X_{\Omega_o}$ in the expectation and finally, recalling the "independence of the past" of a Brownian Motion and that $\mathbb E|W(s)-W(t)|^2= |s-t|$,  we get

$$
\mathbb{E}(\triangle q_{t_i}^b)^2 = \mathbb{E}\lbrace \frac d n + b^2 (q_{t_i}^b,t_i)\frac 1 {n^2}
+ \nabla b (q_{t_i}^b , t_i)\frac 1 {n^2} + o (\frac 1 {n^2})\rbrace^2
$$
Finally, by taking the limit of the Riemannian sums, we have

\begin{equation}\label{rin1}
\lim_{n\rightarrow\infty} n \mathbb{E} \sum_{i=0}^{n-1} (\triangle q_i^b)^2 - nd =
\mathbb{E}\int_0^1 (b^2 (q^b_t , t) + \nabla b (q_t^b, t))dt   \quad\quad
\end{equation}

Moreover, putting $b = v + \frac 1 2 \nabla ln \rho$, recalling  $h_3$) in Definition 2 and the finite action
condition in Definition 1, we can integrate by parts to get

\begin{equation}\label{rin2}
\mathbb{E}\int_0^1 (b^2(q_t^b, t)+\nabla b(q_t^b, t))dt = \int_0^1 \int_{\mathbb{R}^d}
 (v^2 - (\frac 1 2 \nabla ln \rho)^2)\rho dx dt <\infty   \quad\quad
\end{equation}
\end{proof}

\section {Constructing an ancillary convex functional in a fully probabilistic setting }
 Let $(\Omega,\mathcal F,(\mathcal F_t)_t)$  be defined as in Theorem 1 and let $X$ be the configuration process.
 Let  also $\mathbb P$ be a probability measure on $(\Omega,\mathcal F,(\mathcal F_t)_t)$   and $W$ a standard Brownian Motion on $(\Omega,\mathcal F,(\mathcal F_t)_t,\mathbb P)$.
We assume that $\mathbb P$ is such that $X_o$ has probability density equal to $\rho_o$.

We introduce the following space of  stochastic processes
\begin{equation}\label{Ldue}
L^2_{[0,1]}(\mathbb{P}) := \{\beta:\Omega \times [0, 1]\rightarrow\mathbb{R}^d\; s.t.
\int_\Omega \int_0^1\beta^2(\omega,s) ds \mathbb P(d\omega) <\infty\}.
\end{equation}
$L^2_{[0,1]}(P)$ is a Hilbert space with scalar product

$$
<\beta,\gamma>_{L^2_{[0,1]}(\mathbb P)}:= \int_\Omega \int_0^1 \beta (\omega,t)\gamma(\omega,t) dt \mathbb P(d\omega)
$$

\noindent $\beta$ and $\gamma$ belonging to $ L^2_{[0,1]}(\mathbb{P})$. 

We define

 $$
 q_t^\beta := X_0 + \int_0^t \beta_s ds +W_t ,\quad   X_0(\omega)=\omega(0)
 $$
  and
$$
\Delta q_i^\beta := q_{t_{i+1}}^\beta - q_{t_i}^\beta
$$

 Considering the equipartition $\{t_i\}_{i=0}^{n}$ of $[0,1]$ and denoting by $\mathbb E$ the integration with respect to $\mathbb P$, we introduce the elementary convex functional $F_n: L_{[0,1]}^2 (\mathbb{P})\rightarrow \mathbb{R}$

\begin{equation}\label{Fn}
F_n (\beta) :=  n\mathbb E \sum_{i=0}^{n-1} (\triangle q_i^\beta)^2
\end{equation}
One can observe that, dropping the Brownian Motion and reducing the randomness only to the initial conditions, this functional can be seen as an average on the initial configurations of the classical action in discrete time, for a free finite dimensional system.

We say that $\beta$ is a "Markovian drift"  if
$\beta_t = \beta_t^b $ where

\begin{equation}\label{beta}
\beta_t^b:= b (q^b_t,t),
\end{equation}

and \noindent $q^b$ satisfies the stochastic differential equation

$$
q^b_t = X_o+\int _o^t b(q^b_s,s)ds +W_t 
$$

 Recalling the notation

\begin{equation}\label{brhov}
b[\rho,v] := v +\frac 1 2 \nabla \log \rho ,
\end{equation}

\noindent if $b$ is equal to $b[\rho,v]$  with $(\rho,v)\in \Xi$ , then, by the finite action condition in Defintion 1,
 $\beta^b$  belongs to $L^2_{[0,1]}(\mathbb P)$.

Moreover, exploiting  Nelson's renormalization formula (Lemma 2),  we can take the limit for $n$ going to infinity, getting

\begin{equation}
\lim_{n\rightarrow\infty} n \mathbb{E} \sum_{i=0}^{n-1} (\triangle q_i^b)^2 - nd =\int_0^1 \int_{\mathbb{R}^d}
 (v^2 - (\frac 1 2 \nabla ln \rho)^2)\rho dx dt <\infty
\end{equation}
 \noindent where the limit is finite thanks to the finite action condition

Then, for any "Markovian element" $\beta^b$,
$b\equiv b[\rho,v]$, $(\rho,v)$ in $\Xi $, the limit for $n$ going to infinity of $F_n(\beta^b)- nd$ exists and we have

\begin{equation}
\lim_{n\rightarrow\infty} (F_n(\beta^b)- nd)< \infty
\end{equation}

\noindent Moreover, for any $\lambda\in[0,1]$ and $(\beta^{b_1},\beta^{b_2}) \in L^2_{[0,1]}(\mathbb{P})$, $b_1 :=b[\rho_1,v_1]$ and $b_2 :=b[\rho_2,v_2]$ with $(\rho_1,v_1)$ and $ (\rho_2,v_2)$ in $\Xi$,

\begin{multline}
\lim_{n\rightarrow\infty}(F_n(\lambda\beta^{b_1} + (1 - \lambda)\beta^{b_2})- nd)\leq\\
 \leq \lim_{n\rightarrow\infty} \{\lambda (F_n(\beta^{b_1}) - nd) + (1-\lambda)(F_n(\beta^{b_2})- nd)\} < \infty
\end{multline}

Since the convex combination of three elements is equal to the convex combination of proper two elements, one can see by induction that for any finite convex combination
$\sum_{i=1}^m \alpha_i\beta^{b_i}$, $(\beta ^{b_i})_{i=1}^{m}$  being Markovian elements in $L_{[0,1]}^2 (\mathbb{P})$, we have
$$
\underset {n\to \infty} {\lim}[F_n (\sum_{i=1}^m \alpha_i \beta^{b_i}) - nd] < \infty
$$

Denoting by $\Sigma$ the convex  set given by all finite convex combinations of Markovian elements
in $L^2_{[0,1]}(\mathbb{P})$, we can define the convex functional

$$\hat{F}_\infty : L^2_{[0,1]}(\mathbb{P})\to \bar \R $$

\begin{equation}\label{Finfinito}
\hat F_\infty (\beta) := \begin{cases}
                               \underset {n\to \infty}{lim}(F_n(\beta)- nd)\quad &\forall\beta \in \Sigma \\
                               +\infty \quad \text {otherwise}
                             \end{cases}
\end{equation}

The elements of $\Sigma$ are not markovian in general and they are somehow reminiscent of the quantum mixtures, but we must emphasize that describing quantum mixtures would require an enlarged probability space (see for example \cite{CPM}).

The action functionals $A^Q$ and $\hat F_\infty$ satisfy the equality

\begin{equation}\label {equality} 
A^Q (\rho,v) = \hat F_\infty (\beta^{b[\rho,v]})
\end{equation}

\noindent for all $(\rho,v)$ in $\Xi$.

\vskip 10mm

\section{Optimality result}

\vskip 4mm

In this section we show that Lemma 1 and the convexity of $\hat F_\infty$ allows to face the minimisation problem $P_1$ for $A^Q$, which was formulated in Section 2.

\begin{thm}

Let $(\rho_o,\rho_1)$ be probability densities on $\R^d$ with finite variance.

Assume that $\psi:= \rho^{\frac 1 2}\exp^{iS}$ is a solution  in $L^2(\R^d, \mathbb C)$ of the free Schroedinger equation.

\begin{equation}\label{Schroedinger2}
i\partial_t \psi + \frac1 2 \nabla^2 \psi = 0\quad,\quad |\psi_0|^2 = \rho_0 , |\psi_1|^2 = \rho_1,
\end{equation}

\noindent and that $(\rho,\nabla S)$ belongs to $ \Xi(\rho_o,\rho_1)$ as defined in Section 3 , definition 2 . We assume also that  $\psi$ is the unique solution of \eqref {Schroedinger2} in $L^2(\R^d, \mathbb C)$ up to a constant phase and that $\nabla S (x,t)$ be different from zero for all $(x,t)$ belonging to $N_\rho^c$.

Then

\begin{equation}\label{inequality}
A^Q(\rho,\nabla S) \leq A^Q(\rho',v')
\end{equation}

$$
\forall (\rho',v')\in \Xi(\rho_o,\rho_1)\quad   s.t. \quad N_{\rho'} \supseteq N_{\rho} 
$$

\end{thm}

\begin{proof}

Let $(\rho',v')$ belong to $ \Xi(\rho_o,\rho_1)$ and let  

\begin{equation}\label{nodi}
N_{\rho'} \supseteq N_{\rho}
\end{equation}
 
 Put $g:= \rho'-\rho$. Then, by \eqref{nodi},  $g$ is equal to zero on $N_\rho$, because the nodes of $\rho$ are also nodes of $\rho'$,  so that its support is a subset of $N_\rho^c$. Moreover, being $\rho'$ and $\rho$ both normalized to $1$ one has

 \begin{equation}\label{zero}
 \int_{\R^d} g dx=0
 \end{equation}
 
  Then $\rho +yg$ is also normalized to $1$ for all $y\in [0,1]$ . We notice that  $N_{\rho+yg} \supseteq N_{\rho}$ for all $y\in [0,1]$ because, being $g$ equal to zero on $N_\rho$,
 no node of $\rho$ can be delated in the variation.
  
  \vskip 10mm

We consider firstly the case when $g$ is of class $C_K^\infty$.
\vskip 2mm

 (notice: all equations in the following must be understood for all $(x,t)\in N_\rho^c$)

\vskip 2mm
Introducing a time dependent vector field $Z_y^g: N_\rho ^c\to \R ^d$, we consider the family $(\rho_y,v_y)_{y\in [0,1]}$  defined by

\begin{equation}\label{family}
\begin{cases}
\rho_y := \rho +yg\\
v_y := \nabla S + Z_y^g
\end{cases}
\end{equation}

Requiring that  $(\rho_y,v_y)$ satisfies, for all $y\in [0,1]$, the continuity equation

\begin{equation}
\partial _t\rho_y + \nabla\cdot (\rho_y v_y)=0
\end{equation}

\noindent and putting

$$
u:= (\rho +y g)Z_y^g
$$

\noindent one gets
$$
 \nabla \cdot u= -y[\partial _t g+\nabla \cdot (g \nabla S)]
$$

\vskip 2mm

This linear equation in the unknown $u$ can be  easily solved and  the solution depends smoothly on $y$.
 
  Recalling in particular condition \eqref{zero}, one can see that  $Z_y^g$  belongs to $ C_K^\infty (N_\rho^c \to \R^d)$.
 Thus  $(\rho_y,v_y)$ belongs to $ \Xi(\rho_o,\rho_1)$  for all $y\in [0,1]$.

One can also check that  $\frac {d}{dy} X_y^g|_{y=o}$ also belongs to $C_K^\infty (\R^d\times [0,1]\to \R^d)$ and its support is a subset of $N_\rho^c$.

 Defining

$$
\delta \rho := g
$$
 and
$$
\delta v:= \frac {d} {dy} Z_y^g|_{y=0}
$$
\noindent one can see that conditions a) and b) in Lemma 1 are satisfied. Then

\begin{equation}\label{base}
\lim_{y\to 0}\frac 1 y (A^Q(\rho_y,v_y)-A^Q(\rho,\nabla S))=
 D A^Q(\rho,\nabla S)(\delta\rho,\delta v)=0
\end{equation}
\vskip2mm
 
 We can exploit equality \eqref{equality} to get

\begin{equation}\label{uno}
 \underset{y \to 0}{\lim}\frac 1 y( \hat F_\infty (\beta ^{b[\rho_y, v_y]) }) - \hat F_\infty (\beta ^{b[\rho, \nabla S]})) =0
 \end{equation}
 
 \vskip 2mm
 We observe that, since one can see that $b[\rho_y, v_y](x,t)$ depends smoothly on $y$ for all $(x,t)\in N_\rho^c$ by construction, then, by Propositions 3  in Section 4, $\beta ^{b[\rho_y, v_y]}$  also depends smoothly on $y$ and its derivatives have finite momenta. Thus in particular $\frac {d}{dy} \beta ^{b[\rho_y, v_y]) }|_{y=0}$ belongs to $L^2_{[0,1]}(\mathbb P)$.

 We can expand in Taylor series $\beta ^{b[\rho_y, v_y]) }$ with respect to $y$ in \eqref{uno}, getting the equality
 
 $$
 \underset{y \to 0}{\lim}\frac 1 y( \hat F_\infty (\beta ^{b[\rho, \nabla S]} + y\frac {d} {dy} \beta ^{b[\rho_y, v_y]) }|_{y=0}) - \hat F_\infty (\beta ^{b[\rho, \nabla S]}) =0
 $$

 But, being $\hat F_\infty$ convex on any straight line in $L^2_{[0,1]} (\mathbb P)$, $\hat F_\infty (\beta ^{b[\rho_y, v_y]})$ is locally convex as function of $y\in [0,1]$  in a neighborhood of $y=0$.
  
Then $A^Q(\rho_y,v_y) \equiv \hat F_\infty (\beta ^{b[\rho_y, v_y]}) $ has a local minimum for $y=0$.

If $\psi :=\rho^{\frac 1 2} \exp{iS }$ is the unique $L^2$-solution of the free Schroedinger equation \eqref{Schroedinger2}, up to a constant phase, and $\nabla S(x,t)$ is different from zero in every point of $N_\rho^c$, then, by the necessary condition in Lemma 1, the minimum is also global.

Concluding we have

 $$
A^Q(\rho,\nabla S) \leq A^Q(\rho_y,v_y)  \quad \quad  \forall  y\in [0,1]
$$
\noindent and consequently

$$
A^Q(\rho,\nabla S) \leq A^Q(\rho',v')
$$
\noindent where $\rho' -\rho =g$, $g\in C_K^\infty (N_\rho^c, \R)$.

\vskip 20mm 

For a generic element $(\rho',v')$  of $ \Xi (\rho_o,\rho_1)$ put again $g:= \rho'-\rho$ and let $N_{\rho'} \supseteq N_{\rho}$.

 Then g is smooth and such that $\int_{\R^d} g dx =0$ but we do not assme that its support is compact.

 Exploiting again the procedure considered in the proof of Proposition 2,  let  $(D_j)_j $ be a sequence of compact subsets of $N_\rho^c$ such that $D_j\uparrow N_\rho^c$ and define the sequence $( h_j )_j \subset C_K^\infty (\mathbb{R}\times [0,1]\rightarrow \mathbb{R})$ 
where $0\leq h_j \leq 1$ and
\begin{equation}\label{h}
h_j (x,t) = \begin{cases}
                  1     \qquad     (x,t) \in D_j \\
                  0     \qquad     (x,t) \in \underset{m=j+2}{U^\infty} (D_m \setminus D_{m-1})
               \end{cases}
\end{equation}

 Moreover we choose $(h_j)_j$  in such a way that
 $\int_{\R^d} h_j g dx $ is equal to zero for all $j \in \mathbb N $ .

Then, putting $g_j:= gh_j$, $\rho'_j:=\rho + g_j$ and $v'_j:= v+Z_j^{g_j}$, where $Z_y^{g_j}$, $y\in [0,1]$, is constructed as before, one has, for all $j\in \mathbb N$,

$$
A^Q(\rho,\nabla S)\leq A^Q(\rho'_j,v'_j)=
$$
$$
=\int _{D_j}(v'^2-(\frac 1 2 \nabla \log \rho')^2)\rho' dx dt)+\int _{D_j^c}(v_j'^2-(\frac 1 2 \nabla \log \rho'_j)^2)\rho'_j dx dt)
$$

Then
$$
A^Q(\rho,\nabla S)\leq \lim_{j\to \infty} A^Q(\rho'_j,v'_j)= A^Q (\rho',v'))
$$
\noindent for all $(\rho',v')$ in $ \Xi_(\rho_o,\rho_1)$ and such that $N_{\rho'} \supseteq N_{\rho} $.
\end{proof}

\section {Acknowledgements}

The author wish to thank her friend and colleague Paolo Dai Pra for his kind support and valuable comments.


\newpage
\begin {thebibliography} {}




\bibitem{B.B.}B\'enamou J.D. and Y. Br\'enier Y. A computational fluid mechanics solution to the Monge-Kantorovich mass transfer problem, Numerische Mathematik volume 84, 375-393,(2000)


\bibitem{C1}Carlen E. Conservative diffusions, Communications in Mathematical Physics, volume 94, 293-315, (1984)

\bibitem{C3}Carlen E. Progress and problems in Stochastic Mechanics. In: Gielerak,R., Karwoski,W.(eds ) Stochastic Methods in Mathematical Physics, 3-31, Singapore World Scientific (1989)

\bibitem{CPM}Cufaro Petroni N. and Morato L.M. Entangled states in Stochastic Mechanics. Physica A: Math. Gen.{\bf 33}, 5833-5848 (2000)


\bibitem{GM} Guerra F. and Morato L.: Quantization of
Dynamical Systems and Stochastic Control Theory, {\it Phys.Rev.D},
{\bf 27},1774-1786 (1983)



\bibitem{L} Loffredo M.I. Eulerian Variational Principle in Stochastic Mechanics, private communication, (1986), Rapporto Matematico 226, Universita' di Siena (1990)


\bibitem {M2} Morato L.M. Free time evolution of the quantum wave function and optimal transportation, Journal of Mathematical Physics, (2022)


\bibitem{Nelson1}Nelson E. Dynamical  Theories  of  Brownian  Motion (Princeton,  NJ:  Princeton  University      Press)(1967)

\bibitem{Nelson*}Nelson E. in Seminaire de Probabilites,Vol. XIX of lecture notes in Mathematics, edited by J.Azema and M Yor, Springer, New York,(1984)

 \bibitem{Nelson2} Nelson E. {\it  Quantum Fluctuations} (Princeton University Press)(1985)
 
\bibitem{Reddiger} Reddiger M. and Poirier B. Towards a mathematical Theory of the Madelung Equations: Takabayasi's quantization condition, quantum quasi-irrotationality, weak formulations, and the Wallstrom phenomenon, J. Phys. A: Math. Theor. 56, 193001 (2023)

\bibitem{V} Villani C. {\it Topics in Optimal Transportation}, American Mathematical Soc. ISBN     978-0-8218-3312-4,(2003)

\bibitem {Zeng} Zeng W. in Stochastic Processes-Mathematics and Physics ||,S.Albeverio, Ph.Blanchard  and L.Sreit eds. (Springer ),(1985)

\end {thebibliography}
\end {document}